\documentclass[runningheads]{llncs}
\usepackage[T1]{fontenc}
\usepackage{amsmath}
\usepackage{multicol}
\usepackage{amsfonts}
\usepackage{xcolor}
\usepackage{graphicx}
\usepackage{tcolorbox}

\usepackage{tcolorbox}
\usepackage{tikz}
\usepackage[linesnumbered, ruled, vlined]{algorithm2e}
\usepackage{caption}
\usepackage{subcaption}
\usepackage{cite}

\usepackage{url}
\usepackage{hyperref}
\hypersetup{
    colorlinks=true,
    linkcolor=blue,
    filecolor=magenta,      
    urlcolor=cyan,
}
\begin{document}
% \title{Fairness Driven Slot Allocation Problem in Billboard Advertisement}
\title{Fairness Driven Slot Allocation Problem in Billboard Advertisement \thanks{This work is supported by the Start-Up Research Grant provided by the Indian Institute of Technology Jammu, India (Grant No.: SG100047).}}
\author{Dildar Ali\inst{1}  \and Suman Banerjee\inst{1} \and Shweta Jain\inst{2} \and Yamuna Prasad\inst{1}}
\authorrunning{Ali et al.} % abbreviated author list (for running head)
\institute{Indian Institute of Technology Jammu, J \& K-181221, India \email{\{2021rcs2009,suman.banerjee,yamuna.prasad\}@iitjammu.ac.in}\and Indian Institute of Technology Ropar, Rupnagar, Punjab-140001, India  \\
\email{shwetajain@iitrpr.ac.in}}
\maketitle
\begin{abstract}
In \emph{billboard advertisement}, a number of digital billboards are owned by an \emph{influence provider}, and several commercial houses (which we call advertisers) approach the influence provider for a specific number of views of their advertisement content on a payment basis. Though the billboard slot allocation problem has been studied in the literature, this problem still needs to be addressed from a fairness point of view. In this paper, we introduce the \textsc{Fair Billboard Slot Allocation} Problem, where the objective is to allocate a given set of billboard slots among a group of advertisers based on their demands fairly and efficiently. As fairness criteria, we consider the \emph{maximin fair share}, which ensures that each advertiser will receive a subset of slots that maximizes the minimum share for all the advertisers. We have proposed a solution approach that generates an allocation and provides an approximate maximum fair share. The proposed methodology has been analyzed to understand its time and space requirements and a performance guarantee. It has been implemented with real-world trajectory and billboard datasets, and the results have been reported. The results show that the proposed approach leads to a balanced allocation by satisfying the maximin fairness criteria. At the same time, it maximizes the utility of advertisers.
\keywords{Billboard Advertisement, Influence, Slot, Advertiser, Utility, Allocation, Maximin Fair Share}
\end{abstract}
\section{Introduction} \label{Sec:Introduction}
 In recent years, Billboard Advertisement has emerged as an effective out-of-home advertisement technique. Many commercial houses have adopted this mode of advertisement due to its guaranteed return on investment. In a billboard advertisement, there is an \emph{influence provider} who owns a set of digital billboards, and a set of advertisers approaches the influence provider for a specific amount of influence on a payment basis. A few computational problems have been studied in the context of billboard advertisements, and in particular, there are two study directions. The first one concerns giving a billboard database; how can we find a limited number of influential billboard slots? The proposed solution methodologies include a modified greedy approach proposed by Zhang et al. \cite{zhang2018trajectory}, a branch and bound approach proposed by Zhang et al. \cite{zhang2019optimizing}, a divide and conquer approach proposed by Wang et al.\cite{10.1145/3495159}, Pruned Submodularity Graph-based approach, Spatial Clustering based approach proposed by Ali et al. \cite{ali2022influential} The other direction of research aims to minimize the regret while allocating the billboard slots among the multiple advertisers. The proposed solution methodologies for the regret minimization problem include several local search-based techniques proposed by Zhang et al. \cite{zhang2021minimizing}, Budget Effective Greedy and Randomized Greedy approaches proposed by Ali et al. \cite{ali2024minimizing}, and many more.
 
 \textbf{Counter Example 1.} \label{Counter_Example}
Consider an influence provider have six billboard slots $\mathcal{U} = \{ bs_{1}, bs_{2}, bs_{3}, bs_{4}, bs_{5}, bs_{6}\}$ and four advertisers $\mathcal{A} = \{ a_{1}, a_{2}, a_{3}, a_{4}\}$ as reported in Table \ref{ETable:4} and \ref{ETable:5}. Now, the advertisers approach the influence provider with their required influence demand and the corresponding budget shown in Table \ref{ETable:5}. Next, the influence provider's goal is to allocate billboard slots to the advertisers and assume it follows the order be $(a_{1}, a_{2}, a_{3}, a_{4})$ for allocation. Hence, the influence provider will get full payment from the advertiser $a_{1}$ and $a_{2}$ and from the advertiser $a_{3}$; a partial payment will be made due to some penalty for not fulfilling required influence demand i,e., $ \mathcal{I}(\pi_{a_{3}}) < \sigma_{a_{3}}$. However, in the case of advertiser $a_{4}$, no billboard slots remain to allocate from the influence provider side and allocation of $a_{4}$, i.e., $\pi_{a_{4}}$ is empty as shown in Table \ref{ETable:6}. 

\begin{table}[!h]
\vspace{-0.2in} 
\begin{center}
\begin{minipage}{0.48\textwidth}
\small
   \centering
   \begin{tabular}{| c | c | c | c | c | c | c |}
   \hline
   $\mathcal{BS}_{i}$ & $bs_{1}$ & $bs_{2}$ & $bs_{3}$ & $bs_{4}$ & $bs_{5}$ & $bs_{6}$ \\ \hline
   $\mathcal{I}(bs_{i})$ & 2 & 6 & 3 & 7 & 1 & 1 \\ \hline
   Cost & \$4 & \$12 & \$6 & \$14 & \$2 & \$2 \\ \hline
   \end{tabular}
   \caption{\label{ETable:4} Billboard Info.}
\end{minipage}
\hfill
\begin{minipage}{0.48\textwidth}
\small
   \centering
   \begin{tabular}{ | c | c | c | c | c |}
   \hline
   Advertiser($\mathcal{A}$) & $a_{1}$ & $a_{2}$ & $a_{3}$ & $a_{4}$ \\ \hline
   Demand($I_{i}$) & 5 & 7 & 8 & 3 \\ \hline
   Budget($L_{i}$) & \$15 & \$15 & \$17 & \$6 \\ \hline
   \end{tabular}
   \caption{\label{ETable:5} Advertiser Info.}
\end{minipage}
\vspace{0.1cm} % Space between rows
% Second Row: Table 3
\begin{minipage}{0.96\textwidth}
\small
   \centering
   \begin{tabular}{ | c | c | c | c | c |}
   \hline
   $\mathcal{A}$ & $a_{1}$ & $a_{2}$ & $a_{3}$ & $a_{4}$ \\ \hline
   $\mathcal{BS}_{i}$ & $bs_{2}$ & $bs_{4}$ & $bs_{1},bs_{3},bs_{5},bs_{6}$ & - \\ \hline
   $\mathcal{I}(\mathcal{BS}_{i}) - I_{i}$ & 1 & 0 & -1 & 0 \\ \hline
   Satisfied & Yes & Yes & No & No \\ \hline
   \end{tabular}
   \caption{\label{ETable:6} Billboard Slot Allotment.}
\end{minipage}
\end{center}
\vspace{-0.5in} 
\end{table}

\par The proposed solution approaches in the existing literature will return an allocation of slots. However, this allocation may not always be practically usable. Consider the example described in Example \ref{Counter_Example} containing the billboard database, advertiser database, influence demand, and budget of the advertisers. Now, if we apply the Spatial Clustering-based approach proposed by Ali et al. \cite{ali2022influential} on this problem instance, the allocation of slots to the advertisers is shown in Table \ref{ETable:6}. It can be observed that advertiser $a_3$ and $a_{4}$ have not been satisfied. Hence, this allocation may not be considered fair as $a_3$ and $a_{4}$ may envy with $a_1$ or $a_{2}$. Given a set of indivisible goods and a set of agents, it remained a central question how we can fairly divide the goods among the agents for decades. A significant amount of research has been done on this problem. However, it is surprising that the existing literature does not address the issue of fairness in the context of billboard slot allocation. To tackle this problem, in this paper, we pose the slot allocation problem as a problem of fair division of indivisible items, which is a very well-studied problem in Fair Division literature.  In particular, we make the following contributions in this paper:
\begin{itemize}
    \item We introduce the problem of the Fair Billboard Slot Allocation Problem for which no existing literature exists.
    \item In this problem context, we introduce maximin fairness notions and study its properties in our problem context.
    \item We propose an approximation algorithm that provides a fair solution to the billboard slot allocation problem.
    \item We implement the proposed solution methodologies with real-life datasets, and several experiments have been conducted to show the effectiveness and efficiency of our solutions. 
\end{itemize}

The rest of the paper is organized as follows. Section \ref{Sec:BPD} describes required background concepts and defines the problem formally. Section \ref{Sec:PS} describes the proposed solution approach with illustration and analysis. Section \ref{Sec:Experimental_Evaluation} contains the experimental evaluations of the proposed solution approaches. Section \ref{Sec:CFD} concludes this study and gives future research directions.

\section{Background and Problem Definition} \label{Sec:BPD}
\subsection{Notation and Preliminaries}
This section introduces the background of the problem and presents its formal definition. Let us assume, there is a set of $\ell$ billboard $\mathcal{B}=\{b_1, b_2, \ldots, b_\ell\}$. Each one of them are operating for the time interval $[t_1,t_2]$, and we assume that $T = t_{2}- t_{1}$. Consider each billboard is allocated slot-wise for advertising, and each slot has a time duration of $\Delta$.  Assuming that a person $u_i$ crosses a billboard, $b_i$ at a time $t_i$. Now assume that the advertisement content of an E-Commerce house is displayed on that billboard in the slot $[t_x,t_y]$, and $t_i \in [t_x,t_y]$. Then $u_i$ is likely to be influenced by the advertisement content with a certain probability. The billboard $b_i$ will influence the trajectory $t_j$ with probability $Pr(b_i,t_j)$. One of the way to calculate this value as,  $Pr(b_i,u_i) = \frac{Size(b_i)}{\underset{b_{i} \in \mathcal{B}}{max} \ Size(b_{i})} $ where $Size(b_{i})$ is the billboard panel size. We adopt this probability setting \cite{zhang2018trajectory,zhang2019optimizing,zhang2020towards} in our experiments as well. However, it can be calculated in several ways depending on the needs of applications \cite {zhang2020towards,ali2022influential,ali2023influential}. For any positive integer $n$, $[n]$ denotes the $\{1,2,\ldots, n\}$. Next, we discuss the trajectory and billboard database. 

\begin{definition} [Trajectory Database] \label{Def:TD}
A trajectory database $\mathcal{D}$, contains the tuples of the following form: $<u_{id}, \texttt{u\_loc},\texttt{time-slot}>$ where $u_{id}, \texttt{u\_loc}$, and $\texttt{time-slot}$ denotes user id, user location and slot duration, respectively.
\end{definition}

The presence of the tuple $<u_{1}, \texttt{Republic\_Airport},[1200,1400]>$ in the trajectory database $\mathcal{D}$, means that the person with the unique id $u_{1}$ was present at \texttt{Republic\_Airport} during the time interval from $1200$ to $1400$. Next, we describe the Billboard Database in Definition \ref{Def:BD}.
 
\begin{definition}[Billboard Database]\label{Def:BD}
A Billboard database $\mathbb{B}$ is a collection of tuples of the form $<b_{id}, \texttt{b\_{loc}}, \texttt{b\_cost}>$, where $b_{id}, \texttt{b\_{loc}}$, and $\texttt{b\_cost}$ denotes billboard id, billboard location and cost, respectively.
 \end{definition}
If $\mathbb{B}$ contains a tuple $<b_{139}, \texttt{Republic\_Airport}, 25>$ means that billboard $b_{139}$ is already placed in the \texttt{Republic\_Airport} and $\$25$ is the cost for renting this billboard for one slot. Next, we describe the notion of billboard slots.

\begin{definition}[Billboard Slot]\label{Def:BS}
Given billboard database $\mathbb{B}$, billboard slot represent as a tuple of the form $(b_i,[t,t+\Delta])$ where $b_i \in \mathcal{B}$ and $t \in \{T_1, T_1+\Delta+1, T_1+2 \Delta+1, \ldots, \frac{T_2-T_1}{\Delta} - \Delta+1 \}$.
\end{definition}
We denote the set of billboard slots as $\mathbb{BS}$. For any subset $\mathcal{S} \subseteq \mathbb{BS}$, its influence $I(\mathcal{S})$ is computed via the triggering model stated in Definition \ref{IBS} as:  

\begin{definition}[Influence of Billboard Slots]\label{IBS}  
Given a subset of billboard slots $\mathcal{S}$ and a trajectory database $\mathcal{D}$, the influence of $\mathcal{S}$ can be computed using triggering model of influence shown in Equation No. \ref{Eq:Influence}.  
\begin{equation} \label{Eq:Influence}
    \mathcal{I}(\mathcal{S})= \underset{t_j \in \mathcal{D}}{\sum} 1 \ - \ \underset{b_i \in \mathbb{BS}}{\prod} (1- Pr(b_i,t_j))
\end{equation} 
where $Pr(b_i, t_j)$ is the probability of $t_j$ being influenced by $b_i$.  
\end{definition}  

We assume $Pr(b_i, t_j)$ is known for all $b_i \in \mathbb{BS}$ and $t_j \in \mathcal{D}$, following methods from \cite{zhang2020towards,zhang2021minimizing}. Lemma \ref{Lemma:1} states a key property of $\mathcal{I}()$.  

\begin{lemma} \label{Lemma:1}  
The influence function $\mathcal{I}()$ over a trajectory database $\mathbb{D}$ and billboard database $\mathbb{B}$ is non-negative, monotonic, and submodular.  
\end{lemma}  

\subsection{Billboard Advertisement and the Payment Model}
As mentioned previously, a set of $n$ advertisers $\mathcal{A}=\{a_1, a_2, \ldots, a_n\}$ approaches to an influence provider $\mathcal{X}$ for an amount of influence by paying some amount of payment. For any advertiser $a_i \in \mathcal{A}$, the following query is sent to the influence provider: $(a_i, \sigma_i, u_i)$. This signifies that the advertiser $a_i$ is asking for an influence of amount $\sigma_i$ in exchange for the utility $u_i$. In response to this query, the influence provider allocates a subset of the billboard slots $\pi_{i} \subseteq \mathbb{BS}$. In return, the advertiser $a_i$ gives a payment $\mathcal{P}_{i}$ to the influence provider. Now, we define the `Allocation of Billboard Slots', and `The Payment Model' in Definition \ref{Def:Allocation} and \ref{Def:Paymet}, respectively. 
\begin{definition} [Allocation of Billboard Slots] \label{Def:Allocation}
    Given a set of billboard slots $\mathbb{BS}$, an allocation of these slots among $n$ advertisers is given as $\Pi = (\pi_1, \pi_2, \ldots, \pi_n)$ where $\pi_i$ contains the billboard slots allocated to the advertiser $a_{i}$. An allocation must follow the following properties:
    \begin{itemize}
        \item for all $i,j \in [n]$ and $i \neq j$, $\pi_i \cap \pi_j = \emptyset$.
        and  $\underset{i \in [n]}{\bigcup} \ \pi_{i}=\mathbb{BS}$ 
    \end{itemize}
\end{definition}
% The outcome of an allocation is an \emph{allocation vector} $\mathbb{A}$ of size $|\mathbb{BS}|$ whose each entry can be from $[1,n]$ and a \emph{payment vector} $\mathcal{P}$ of dimension $n$ whose each entry is a real number. If $\mathbb{A}(i)=k$ signifies that the $i^{th}$ slot has been allocated to the advertiser $a_k$, and if $\mathcal{P}(a_i)=x$, this means that the $x$ amount of payment needs to be made to the influence provider by the advertiser $a_i$.
\begin{definition} [The Payment Model] \label{Def:Paymet}
    For any advertiser $a_i \in \mathcal{A}$, depending upon the obtained influence from the influence provider, the payment of amount $\mathcal{P}(a_i)$ will be made to the influence provider and the following conditional equation defines it. 
 \[
    \mathcal{P}(a_i)= 
\begin{cases}
    u_i,& \text{if } I(\pi_i)\geq \sigma_i \\
    u_i \cdot (1 - \gamma \cdot \frac{I(\pi_i)}{\sigma_i}),              & \text{otherwise}
\end{cases}
\]
In the above equation, the quantity $\frac{I(\pi_i)}{\sigma_i}$ refers to the fraction of satisfied influence by required influence, and $\gamma$ is the penalty ratio whose value lies in $[0,1]$, due to the unsatisfied influence demand.
\end{definition}
Now, we define the utility of an advertiser as follows:
\begin{definition} [Utility of an Advertiser] \label{Def:Utility}
    For an advertiser $a_i \in \mathcal{A}$, let $v_i$ denote its monetary valuation per unit influence. Hence, the utility of the advertiser $a_i$ is defined as $Utility(a_i)= v_i \cdot I(\pi_i) - \mathcal{P}(a_i)$
\end{definition}
 \begin{lemma} \label{Lemma:monotone}
The utility function $Utility(.)$ is monotone by nature.
 \end{lemma} 

% In Definition \ref{Def:Utility}, the utility of advertiser \(a_i\) is defined for its own share \( \pi_i \) as \( Utility_{(a_i)}(\pi_i) \). By default, \( Utility_{(a_i)} \) refers to this utility, omitting \( \pi_i \). We now define the maximin fair share. 
\subsection{Maximin Fair Share}
% \paragraph{\textbf{Max-Min Fair Share.}}
The maximin fair share of an advertiser $a_{i}$, with the influence function $\mathcal{I}()$ is defined as the maximum over all the allocations, the value of the worst share that $a_{i}$ get according to his allocations. Formally, we defined the notion of maximin fair in Definition \ref{max-min}.
\begin{definition}[Maximin Fair Share]\label{max-min}
Let $\mathcal{A}$ be a set of advertisers, $\mathbb{BS}$ be a set of billboard slots, $\mathcal{U} = \{u_{1},u_{2},\ldots u_{n}\}$ be the set of the budget each advertiser spent for their required influence for billboard slots, $\Pi = (\pi_{1},\pi_{2},\ldots,\pi_{n})$ be an allocation for $\mathcal{A}$ and $\mathbb{BS}$. Now, for the advertiser $a_{i}$ maximin fair share can be denoted as $\mathcal{M}_{i}$ and defined as,
\begin{equation} \label{Eq:Eq_max-min}
  v_{i} \cdot \mathcal{I}(\pi_{i}) - \mathcal{P}(a_i) \geq \max_{\Pi} \min_{j} v_{i} \cdot \mathcal{I}(\pi_{j}) - \mathcal{P}(a_j), \forall~ i, j \in [n] ~and~i \neq j.
\end{equation}
We say an allocation $\Pi$ satisfies the max-min fair share criterion if for all $a_{i} \in \mathcal{A}$ satisfy Equation No. \ref{Eq:Eq_max-min}. 
\end{definition}
\subsection{Problem Definition}
Formally, we call our problem a Fair Billboard Slot Allocation Problem (FBSA Problem). This problem is solved from the influence provider perspective by allocating the billboard slots to the advertisers. The input to this problem is the campaign quotation for all the advertisers: $(a_i, \sigma_i, u_i)_{i=1}^{n}$ and the access to the trajectory $\mathcal{D}$ and billboard database $\mathbb{B}$. An arbitrary instance of the FBSA problem is denoted as $I=<(a_i, \sigma_i, u_i)_{i=1}^{n}, \mathcal{D}, \mathbb{B}>$. Given an input instance, the FBSA Problem asks to partition the billboard slots to allocate to the advertisers so that the defined fairness criteria are satisfied. The output to this problem is $(\Pi, \mathcal{P})$ where $\Pi$ is the allocation and $\mathcal{P}$ is the payment vector. From the computational point of view, the FBSA Problem can be represented as follows: 

\begin{tcolorbox}
\underline{\textbf{FBSA Problem}}\\
\textbf{Input:} The campaign proposal of the advertisers $(a_i, \sigma_i, u_i)_{i=1}^{n}$, the trajectory database $\mathcal{D}$, and the billboard database $\mathbb{B}$. \\
\textbf{Problem:} Create an allocation $\Pi=\{\pi_1, \pi_2, \ldots, \pi_n\}$ such that the max-min fair share is satisfied.
\vspace{0.2 cm}\\
\textbf{Output:} $\Pi=\{\pi_1, \pi_2, \ldots, \pi_n\}$ and $\mathcal{P}=\{ \mathcal{P}(a_1), \mathcal{P}(a_2), \ldots, \mathcal{P}(a_n)\}$.
\end{tcolorbox}

\section{Proposed Solution Approaches} \label{Sec:PS}
In this section, we show the existence of a $\frac{1}{3}(1 - \frac{1}{e}-\epsilon)$-approximate maximin fair allocation if the advertisers have a submodular valuation, in polynomial time. We know that finding the maximin share is an NP-hard problem \cite{barman2018finding,10.1145/3033274.3085136}. Hence, instead of using actual maximin share $\mathcal{M}_{a_i}$, we execute Algorithm \ref{alg:RR_submodular} with threshold $\delta_{a_i}$ for all $a_{i} \in \mathcal{A}$. The round-robin approach in Algorithm \ref{alg:RR_submodular} takes $\delta_{i}s$ as a threshold that allocates highly influential billboard slots (w.r.t $\delta_{i}$) as a singleton bundle and distributes the remaining slots to the advertisers in a round-robin manner. The main contribution is to show for each advertiser $a_{i}$, if the $\delta_{a_i} \leq \mathcal{M}_{a_i}$, the bundle $\pi_{i}$ is allocated to $a_{i}$ by Algorithm \ref{alg:RR_submodular} and satisfies $\mathcal{I}^{i}(\pi_{i}) \geq \frac{1}{3}(1 - \frac{1}{e}-\epsilon)\cdot \delta_{a_i}$. It is relevant to note that this guarantee independently holds for all the advertisers as long as $\delta_{a_i} \leq \mathcal{M}_{a_i}$. Formally, we established this guarantee in Theorem \ref{Th:1}. Next, Algorithm \ref{alg:Submodular} starts by initializing thresholds $\delta_{a_i}$s to be greater than the maximum share of the advertisers. We consider that the advertiser whose maximum share is zero has been removed. If the allocation provided by Algorithm \ref{alg:RR_submodular} does not satisfy $\mathcal{I}^{i}(\pi_{i}) \geq \frac{1}{3}(1 - \frac{1}{e}-\epsilon)\cdot \delta_{a_i}$, it decreases the threshold value for all the advertisers. This geometric decrease of the thresholds is justified as previously we initialized the threshold $\delta_{a_i} > \mathcal{M}_{a_i}$. This claim is stated in Lemma \ref{Lemma:1}.
\vspace{-0.1in}
\begin{algorithm}[!h]
\scriptsize
\caption{RoundRobin Approach for Submodular Valuation}\label{alg:RR_submodular}
\SetAlgoLined
\KwData{An instance over $m$ billboard slots and $n$ advertisers, with influence Function $\mathcal{I}()$.}
\KwResult{Allocation $\Pi = (\pi_1, \pi_2, \dots, \pi_n)$ such that $\mathcal{I}^{i}(\pi_{i}) \geq \frac{1}{3}(1 - \frac{1}{e} - \epsilon)\cdot \delta_{a_i}$, for an advertiser $a_i \in [n]$ which satisfies $\delta_{a_i} \leq \mathcal{M}_{a_i}$}.
Initialize set of advertisers $\mathcal{A} = [n]$ and set of slots $\mathcal{B} = [m]$\;
\While{$\text{there exist advertiser~} a_i \in \mathcal{A}$ \text{and slot~} $j \in \mathcal{B}$ \text{such that $\mathcal{I}^{i}(j) \geq \frac{1}{3}(1 - \frac{1}{e}-\epsilon)\cdot \delta_{a_i}$}}{
$\pi_{i} \leftarrow \{j\}$\;
$\mathcal{A} \leftarrow \mathcal{A} \setminus a_{i}$\;
$\mathcal{B} \leftarrow \mathcal{B} \setminus \{j\}$\;
}
Assume via reindexing, that the set of remaining advertisers $\mathcal{A} = \{a_{1},a_{2}, \ldots ,|\mathcal{A}|\}$\;
\While{$\mathcal{B} \neq \emptyset$}{
\For{$i = 1 \text{~to~} |\mathcal{A}|$}{
$\mathcal{R} \leftarrow$ a random subset obtained by sampling $\frac{|\mathcal{B}|}{|\mathcal{A}|} \log \frac{1}{\epsilon}$ random elements from $\mathcal{B} \setminus \mathcal{A}$\;
Pick $g^* \in \arg\max_{j \in \mathcal{R}} \mathcal{I}^{i}(\pi_{i} \cup \{j\}) - \mathcal{I}^{i}(\pi_{i})$\;
$\pi_i \gets \pi_{i} \cup \{g^*\}$\;
$\mathcal{B} \gets \mathcal{B} \setminus \{g^*\}$\;
}}
\Return Allocation $\Pi = (\pi_1, \pi_2, \dots, \pi_n)$
\end{algorithm}
\vspace{-0.1in}

\paragraph{\textbf{Complexity Analysis.}} Now, we analyze the time and space requirement of Algorithm \ref{alg:RR_submodular}. In-Line No. $1$, initialization of advertisers and slots will take $\mathcal{O}(n)$ time. In-Line, No. $2$ \texttt{while loop} will execute for $\mathcal{O}(n)$ time. Hence, Line no $2$ to $6$ will take $\mathcal{O}(n)$ in the worst case. Re-indexing the advertiser will take  $\mathcal{O}(n)$ time. Next, Line No. $8$ and $9$ the \texttt{while loop} and \texttt{for loop} will execute for $\mathcal{O}(m)$ and $\mathcal{O}(n)$ time, respectively. In-Line No. $10$ sampling $\frac{|\mathcal{B}|}{|\mathcal{A}|} \log \frac{1}{\epsilon}$ many slots will take $\mathcal{O}( m \cdot n \cdot \frac{|\mathcal{B}|}{|\mathcal{A}|} \log \frac{1}{\epsilon})$ and Line No. $11$ calculate influence will take $\mathcal{O}(m^{2} \cdot n \cdot |\mathcal{R}| \cdot t)$, where $t$ is the number of tuple in the trajectory database. So, Line No. $8$ to $15$ will take $\mathcal{O}(m \cdot n \cdot \frac{|\mathcal{B}|}{|\mathcal{A}|} \log \frac{1}{\epsilon} + m^{2} \cdot n \cdot |\mathcal{R}| \cdot t + m \cdot n)$ i.e.,  $\mathcal{O}(m \cdot n \cdot \frac{|\mathcal{B}|}{|\mathcal{A}|} \log \frac{1}{\epsilon} + m^{2} \cdot n \cdot |\mathcal{R}| \cdot t)$. Now, the additional space requirement to store advertisers, slots, and allocation will take $\mathcal{O}(n)$, $\mathcal{O}(m)$, and $\mathcal{O}(n)$, respectively. Hence, the total space requirement is $\mathcal{O}(m+n)$.

\begin{theorem}
The time and space requirement for Algorithm \ref{alg:RR_submodular} will be $\mathcal{O}(m \cdot n \cdot \frac{|\mathcal{B}|}{|\mathcal{A}|} \log \frac{1}{\epsilon} + m \cdot n \cdot |\mathcal{R}| \cdot t)$ and $\mathcal{O}(m+n)$, respectively. 
\end{theorem}

\vspace{-0.3in}
\begin{algorithm}[!h]
\scriptsize
\caption{Approximate Maximin Fair Allocation for Submodular Valuation}\label{alg:Submodular}
\SetAlgoLined
\KwData{An instance over $m$ billboard slots and $n$ advertisers, with influence Function $\mathcal{I}()$.}
\KwResult{An allocation $\Pi = (\pi_1, \pi_2, \dots, \pi_n)$ such that, $\mathcal{I}^{i}(\pi_{i}) \geq \frac{1}{3}(1 - \frac{1}{e}- \epsilon)\cdot \mathcal{M}_{a_i}$ for all $a_i \in [n].$ Here, $\lambda \in (0,1)$ is arbitrary small constant.}
For all $a_i \in [n]$, initialize $\delta_{a_i} =\mathcal{I}^{i}([m])$ and $\pi_{i} = \emptyset$\;
Initialize $\mathcal{Y} = \{a_i \in [n] ~|~ \mathcal{I}^{i}(\pi_{i}) < \frac{1}{3}(1 - \frac{1}{e}- \epsilon)\cdot \delta_{a_i}\}$ i.e., $\mathcal{Y} = [n]$\;
\While{$\mathcal{Y} \neq \emptyset$}{
For all $a_i \in \mathcal{Y}$, update $\delta_{a_i} \leftarrow \frac{1}{1+\lambda}\cdot \delta_{a_i}$\;
Update the allocation by executing Algorithm \ref{alg:RR_submodular} with current $\delta$ value: $(\pi_1, \pi_2, \dots, \pi_n) \leftarrow RoundRobin(\delta_{1}, \delta_{2}, \ldots, \delta_{n})$\;
Update $\mathcal{Y} = \{a_i \in [n] ~|~ \mathcal{I}^{i}(\pi_{i}) < \frac{1}{3}(1 - \frac{1}{e}-\epsilon)\cdot \delta_{a_i}\}$\;
}
\Return Allocation $\Pi = (\pi_1, \pi_2, \dots, \pi_n)$
\end{algorithm}
\vspace{-0.3in}
\paragraph{\textbf{Complexity Analysis.}} Now, we analyze the time and space requirements for Algorithm \ref{alg:Submodular}. In-Line No. $1$, for initialization of $\delta$, $\pi$  will take $\mathcal{O}(n)$ and $\mathcal{O}(n)$, respectively. In-Line No. $2$, initialization of $\mathcal{Y}$ will take $\mathcal{O}(m \cdot n \cdot t)$, where $t$ is the number of tuples in the trajectory database. In-Line No. $3$, \texttt{while loop} will run for $\mathcal{O}(n)$ and updating the $\delta$ value in Line No. $4$ will take $\mathcal{O}(n^{2})$ time. In-Line No. $5$ the round-robin will take $\mathcal{O}(m \cdot n^{2} \cdot \frac{|\mathcal{B}|}{|\mathcal{A}|} \log \frac{1}{\epsilon} + m \cdot n^{2} \cdot |\mathcal{R}| \cdot t)$ time and updating $\mathcal{Y}$ will take $\mathcal{O}(m \cdot n^{2} \cdot t)$ time. Hence, Line No. $3$ to $7$ will take $\mathcal{O}(n^{2} + m \cdot n^{2} \cdot \frac{|\mathcal{B}|}{|\mathcal{A}|} \log \frac{1}{\epsilon} + m \cdot n^{2} \cdot |\mathcal{R}| \cdot t + m \cdot n^{2} \cdot t)$ time. So, the total time taken by Algorithm \ref{alg:Submodular} will be $\mathcal{O}(n + m \cdot n \cdot t + m \cdot n^{2} \cdot \frac{|\mathcal{B}|}{|\mathcal{A}|} \log \frac{1}{\epsilon} + m \cdot n^{2} \cdot |\mathcal{R}| \cdot t + m \cdot n^{2} \cdot t)$ i.e.,  $\mathcal{O}(m \cdot n^{2} \cdot \frac{|\mathcal{B}|}{|\mathcal{A}|} \log \frac{1}{\epsilon} + m \cdot n^{2} \cdot |\mathcal{R}| \cdot t)$ time to execute. Next, the additional space requirements will be $\mathcal{O}(m+n)$ to store slots and advertisers.
\begin{theorem}
The time and space requirement for Algorithm \ref{alg:Submodular} will be  $\mathcal{O}(m \cdot n^{2} \cdot \frac{|\mathcal{B}|}{|\mathcal{A}|} \log \frac{1}{\epsilon} + m \cdot n^{2} \cdot |\mathcal{R}| \cdot t)$ and $\mathcal{O}(m+n)$, respectively. 
\end{theorem}

\begin{lemma}\label{Lemma:maximin}
Consider a scenario with $m$ billboard slots and $n$ advertisers, where each advertiser $a_{i}$ has a valuation function $\mathcal{I}^i: 2^{[m]} \to \mathbb{R}_+$ that is nonnegative, monotone, and submodular for $1 \leq i \leq n$. Let $\Pi = (\pi_1, \pi_2, \dots, \pi_n)$ represent an allocation produced by Algorithm~\ref{alg:RR_submodular} with input thresholds $\delta_i \in \mathbb{R}_+$ for all $a_i$. If the thresholds satisfy $\delta_i \leq \mathcal{M}_{a_i}$, then the allocation guarantees:
$\mathcal{I}^i(\pi_i) \geq \frac{1}{3} \left(1 - \frac{1}{e} - \epsilon\right) \cdot \delta_{a_i} \quad \text{for all } a_i.$
\end{lemma}

\begin{theorem}\label{Th:1}
Consider a setting with $m$ billboard slots and $n$ advertisers, where each advertiser $a_{i}$ has a valuation function $\mathcal{I}^i: 2^{[m]} \to \mathbb{R}_+$ that is nonnegative, monotone, and submodular for $1 \leq i \leq n$. Algorithm~\ref{alg:Submodular} computes an allocation $\Pi = (\pi_1, \pi_2, \dots, \pi_n)$ that satisfies $\mathcal{I}^i(\pi_i) \geq \frac{1}{3} \left(1 - \frac{1}{e} - \epsilon \right) \cdot \mathcal{M}_{a_i}, \quad \forall~ a_i \in [n],$ and it does so in polynomial time.
\end{theorem}

\begin{proof}
In Algorithm~\ref{alg:Submodular}, the initial threshold for each advertiser $a_{i}$ is set to $\delta_i = \mathcal{I}^i([m])$, which is guaranteed to be at least the advertiser’s maximin share, $\mathcal{M}_{a_i}$. For each advertiser $a_{i}$, Algorithm \ref{alg:Submodular} never decrement $\delta_{a_i}$ below $\frac{1}{1+\lambda}\cdot \mathcal{M}_{a_i}$ and when Algorithm \ref{alg:Submodular} terminates, for every advertiser $a_{i}$, the $\delta_{a_i}$ satisfies $\delta_{a_i} \geq \frac{1}{1+\lambda}\cdot \mathcal{M}_{a_i}$ and $\mathcal{I}^{i}(\pi_{i}) \geq \frac{1}{3}(1 - \frac{1}{e}-\epsilon)\cdot \mathcal{M}_{a_i}$ since $\mathcal{Y} = \emptyset$ at termination. Hence, Algorithm~\ref{alg:Submodular} achieves the desired approximate fairness guarantee. We tighten the analysis of Lemma \ref{Lemma:maximin} to obtain the approximation guarantee better than $1/3$. With a constant $\lambda \in (0,1)$ we achieve the approximation guarantee of $1/3$ (as reported in Theorem \ref{Th:1}) instead of slightly worst bound $\frac{1}{10(1+\lambda)}$. Furthermore, we can bound the running time of Algorithm \ref{alg:Submodular} by initializing $\delta_{a_i} = \mathcal{I}^{i}([m])$ and As stated in Section \ref{Sec:PS} in Algorithm \ref{alg:Submodular}, the advertisers whose maximin share is zero is removed. The maximum number of times advertiser $a_{i}$, can reside in set $\mathcal{Y}$ is $\log_{(1+\lambda)}(\frac{ \mathcal{I}^{i}([m])}{\mathcal{M}_{a_i}})$. Hence, the overall bound ensures that Algorithm \ref{alg:Submodular} will run in polynomial time.
\end{proof}
\section{Experimental Evaluation}\label{Sec:Experimental_Evaluation}
\paragraph{\textbf{Dataset Description and Setup.}}
We consider the datasets used in the existing literature \cite{ali2022influential,ali2023influential,10.1145/3605098.3636052} are  New York City (NYC)\footnote{\url{https://www.nyc.gov/site/tlc/about/tlc-trip-record-data.page}} and Los Angeles (LA)\footnote{\url{https://github.com/Ibtihal-Alablani}}.The NYC dataset contains 227,428 check-ins recorded from April 12, 2012, to February 16, 2013, while the LA dataset includes 74,170 check-ins collected from 15 specific streets in Los Angeles. In addition, billboard data was obtained from various locations in NYC and LA through LAMAR\footnote{\url{http://www.lamar.com/InventoryBrowser}}, a billboard provider. The billboard dataset for NYC comprises 1,031,040 slots, whereas the dataset for LA contains 2,135,520 slots. An HP Z4 workstation with 64 GB of RAM and an Xeon(R) 3.50 GHz CPU runs all of the Python code of our experiments.

\paragraph{\textbf{Key Parameters.}}
All the key parameters are summarized in Table \ref{Key-parameters}. First, the Demand Supply Ratio ($\alpha$) is the ratio of global influence demand to supply, \(\alpha = \sigma^{\mathcal{T}} / \sigma^{*}\), where $\sigma^{\mathcal{A}} = \sum_{i=1}^{k} \sigma_{i}$ and $\sigma^{*} = \sum_{b \in \mathcal{BS}} \mathcal{I}(b)$. Second, the Average Individual Demand Ratio ($\beta$) is the ratio of average individual demand to the influence supply, $\beta = \sigma^{\mathcal{A}^{''}} / \sigma^{*}$  where $\sigma^{\mathcal{A}^{''}} = \sigma^{\mathcal{A}} / |\mathcal{A}|$ is the average individual influence demand. Third, the Advertiser Demand ($\mathcal{I}$) Once the $\beta$ value is fixed, $\sigma^{\mathcal{A}^{''}}$ can be easily derived as $ \sigma^{\mathcal{A}^{''}} = \beta \cdot \sigma^{*}$. Subsequently, we can generate the demand of each advertiser as $\mathcal{I}_{i} = \lfloor \omega \cdot \sigma^{*} \cdot \beta \rfloor$, where $\omega \in [0.8,1.2]$. Fourth, Following the existing literature \cite{zhang2021minimizing,zhang2019optimizing,zhang2020towards}, we model the slot cost as proportional to its influence: \( Cost(bs) = \left\lfloor \tau \times \frac{\mathcal{I}(bs)}{10} \right\rfloor \), where \(\tau \in [0.9, 1.1]\). Fifth, Following existing literature \cite{zhang2021minimizing,ali2024minimizing}, advertisers payments are set as $u_{i} = \lfloor \psi \cdot \sigma_{i} \rfloor$, where $\psi \in [0.9, 1.1]$, with the advertiser's budget $\mathcal{B} = \sum_{i=1}^{n} u_{i}$. Sixth, we vary the accuracy speed-up parameter, $\epsilon$, from $0.1$ to $0.9$ to sample out slots from larger slots set. We vary $\theta$ from $25m$ to $150m$, and $\theta$ denotes the distance billboard slots can influence trajectories. In every set of experiments, we vary only one parameter and set the remaining as the default setting (highlighted in bold).
\vspace{-0.3in}
\begin{table}[h!]
\caption{\label{Key-parameters} Key Parameters}
\vspace{-0.15 in}
\begin{center}
    \begin{tabular}{ | p{2cm}| p{5.5cm}|}
    \hline
    Parameter & Values  \\ \hline
    $\alpha$ & $40\%, 60\%, 80\%, \textbf{100\%}, 120\%$   \\ \hline
    $\beta$ & $1\%, 2\%, \textbf{5\%}, 10\%, 20\%$   \\ \hline
    $\gamma$ & $0, 0.25, \textbf{0.5}, 0.75, 1\%$   \\ \hline
    $\epsilon$ & $0.1, \textbf{0.3}, 0.5, 0.7, 0.9$   \\ \hline
    $\theta$ & $25m,50m,\textbf{100m},125m,150m$  \\ \hline
    \end{tabular}
\end{center}
\vspace{-0.3in}
\end{table}

% \paragraph{\textbf{Demand Supply Ratio $\alpha$.}}
% The ratio of global influence demand to supply, \(\alpha = \sigma^{\mathcal{T}} / \sigma^{*}\), where $\sigma^{\mathcal{A}} = \sum_{i=1}^{k} \sigma_{i}$ and $\sigma^{*} = \sum_{b \in \mathcal{BS}} \mathcal{I}(b)$ is evaluated for \(\alpha = 40\%, 60\%, 80\%, 100\%,\) and \(120\%\). 
% \paragraph{\textbf{Average Individual Demand Ratio $\beta$.}} 
% The ratio of average individual demand to the influence supply, $\beta = \sigma^{\mathcal{A}^{''}} / \sigma^{*}$  where $\sigma^{\mathcal{A}^{''}} = \sigma^{\mathcal{A}} / |\mathcal{A}|$ is the average individual influence demand.
% \paragraph{\textbf{Advertiser Demand} $\mathcal{I}$.}
% Once the $\beta$ value is fixed, $\sigma^{\mathcal{A}^{''}}$ can be easily derived as $ \sigma^{\mathcal{A}^{''}} = \beta \cdot \sigma^{*}$. Subsequently, we can generate the demand of each advertiser as $\mathcal{I}_{i} = \lfloor \omega \cdot \sigma^{*} \cdot \beta \rfloor$, where $\omega \in [0.8,1.2]$.

% \paragraph{\textbf{Advertiser Budget} $\mathcal{U}$.}
% Following \cite{zhang2021minimizing,ali2024minimizing,ali2024regret}, advertisers payments are set as $u_{i} = \lfloor \psi \cdot \sigma_{i} \rfloor$, where $\psi \in [0.9, 1.1]$, with the advertiser's budget $\mathcal{B} = \sum_{i=1}^{n} u_{i}$.
% \paragraph{\textbf{Environment Setup.}}

\paragraph{\textbf{Baseline Methods.}} Now, we will discuss the baseline methods used in our experiment to compare with the proposed approach as follows.
\paragraph{\textbf{Simple Greedy Allocation.}}
In simple greedy allocation, billboard slots are allocated to the advertiser one by one based on the marginal gain computation till the influence and budget demand is satisfied.
% \paragraph{\textbf{Round Robin Allocation.}}
% In this allocation scheme, billboard slots are assigned cyclically to advertisers based on marginal gains, looping back to the first advertiser after reaching the end of the list until the budget and influence demand are allocated.
\paragraph{\textbf{Random Allocation.}}
Billboard slots are randomly selected to allocate to the advertiser until the influence demand and the budget constraint are satisfied.
\paragraph{\textbf{Top-$k$ Allocation.}}
Billboard slots are sorted based on their influence value in this approach. Next, sorted slots are allocated to the advertiser till the influence demand and the budget constraints are satisfied.

\paragraph{\textbf{Goals of our Experiments.}} \label{Sec:Research_Questions}
The following research questions (RQ) are our focus in this study.
\begin{itemize}
\item \textbf{RQ1}: Varying $\alpha$, $\beta$ how the utility of the advertisers varies.
\item \textbf{RQ2}: Varying $\alpha$, $\beta$ how the number of satisfied advertisers varies.
\item \textbf{RQ3}: Varying $\alpha$, $\beta$ and $\epsilon$ how the computational time varies.
\end{itemize}

\paragraph{\textbf{Experimental Results with Discussions.}}
Now, we discuss the experimental results of the proposed solution methodologies and address the research questions mentioned in the research questions.
% \vspace{-0.25cm}
\begin{figure*}[!ht]
\centering
\begin{tabular}{ccc}
\includegraphics[scale=0.175]{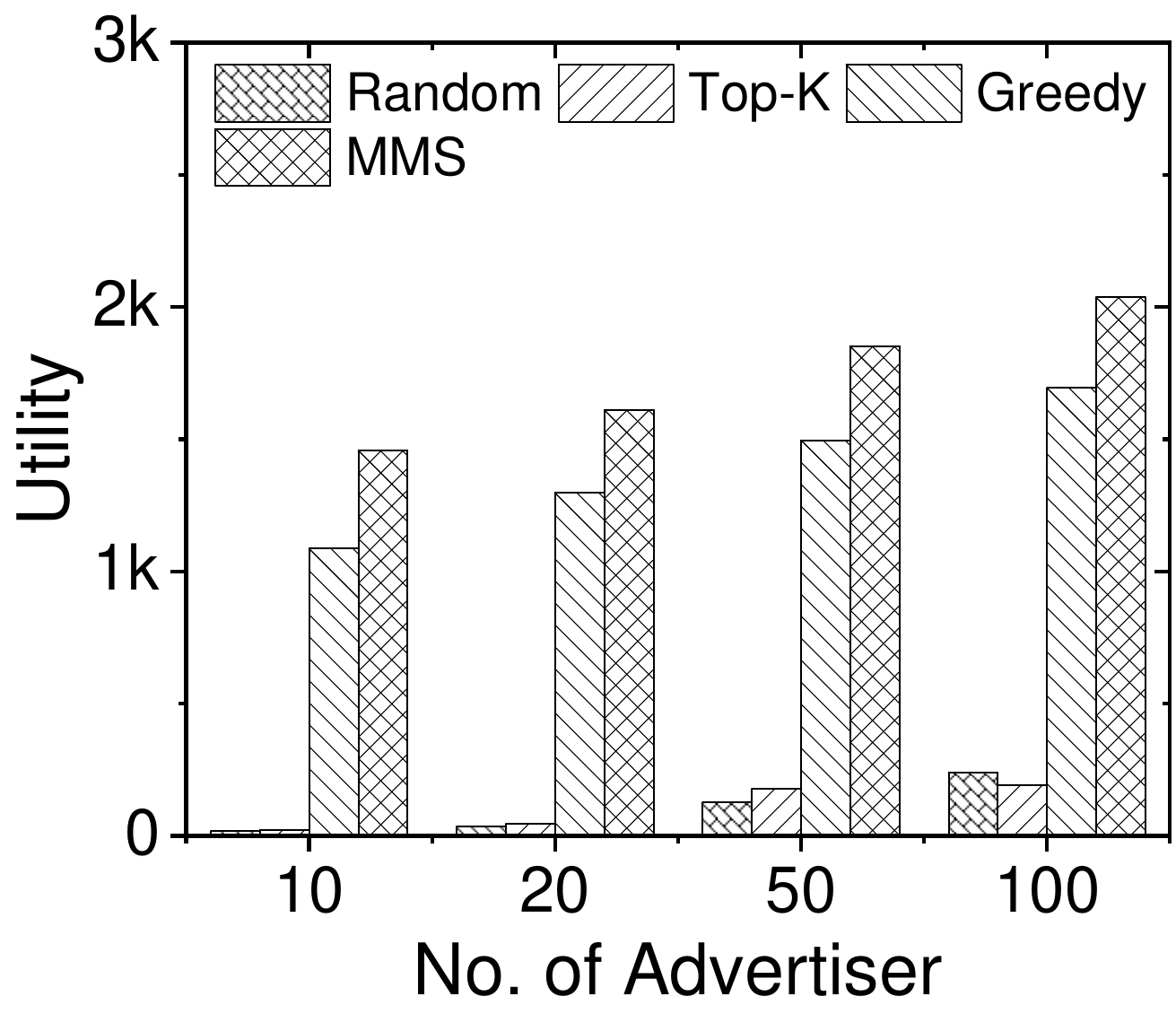} & \includegraphics[scale=0.175]{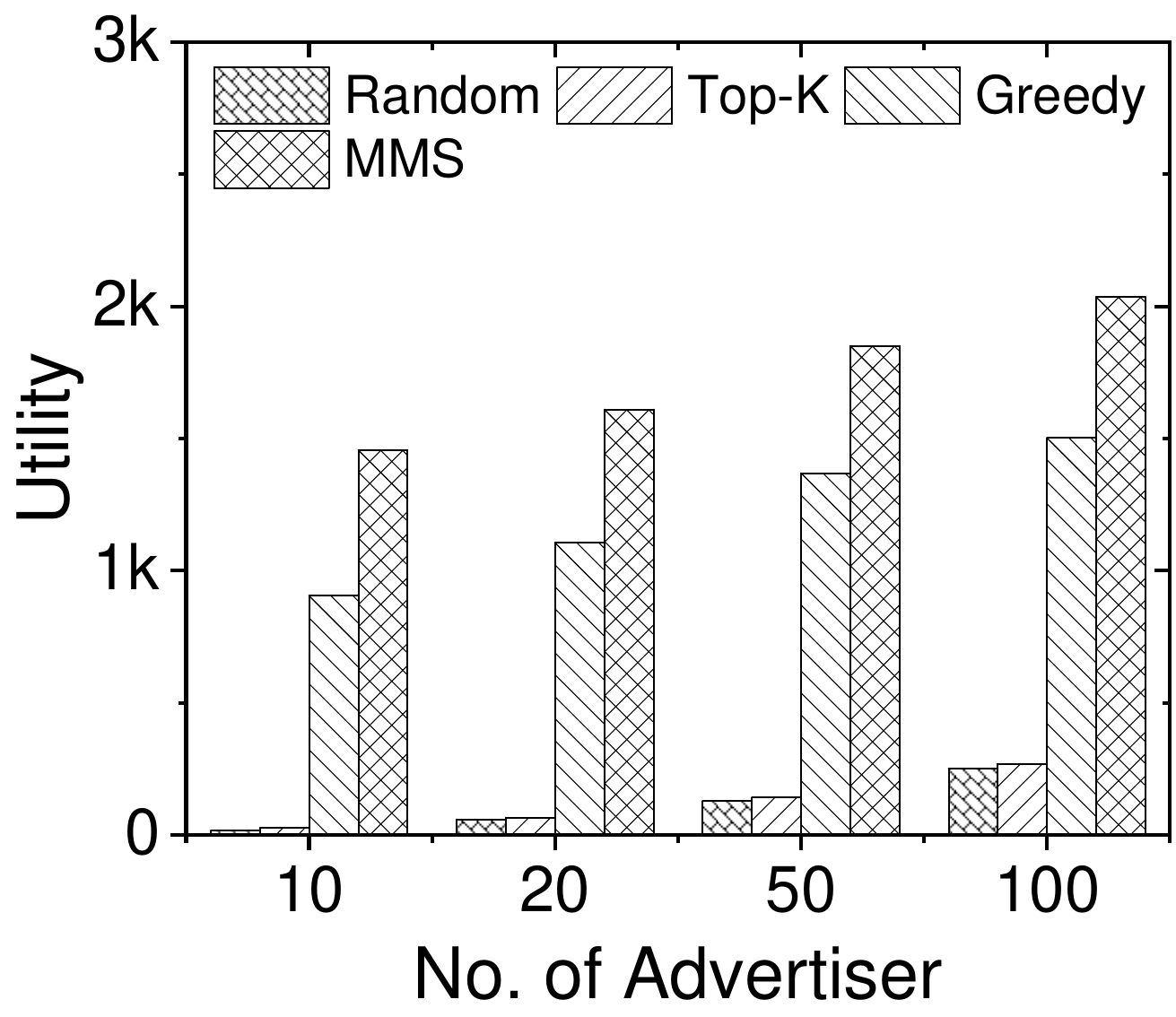} & \includegraphics[scale=0.175]{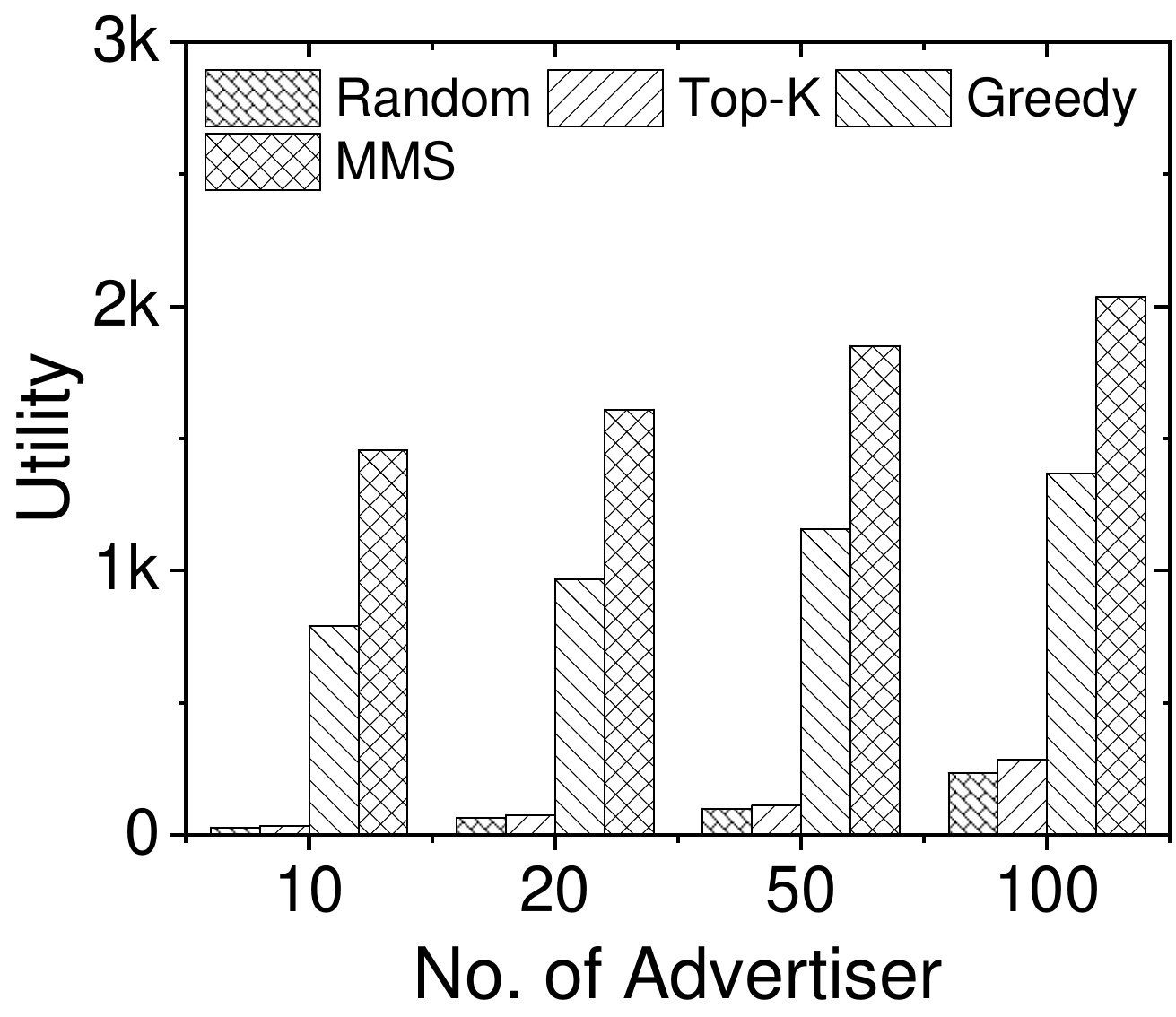} \\
\tiny{(a) Varying $|\mathcal{A}|, \beta ~\text{with}~ \alpha = 40\%$} &  \tiny{(b) Varying $|\mathcal{A}|, \beta ~\text{with}~ \alpha = 60\%$} & \tiny{(c) Varying $|\mathcal{A}|, \beta ~\text{with}~ \alpha = 80\%$}  \\
\includegraphics[scale=0.175]{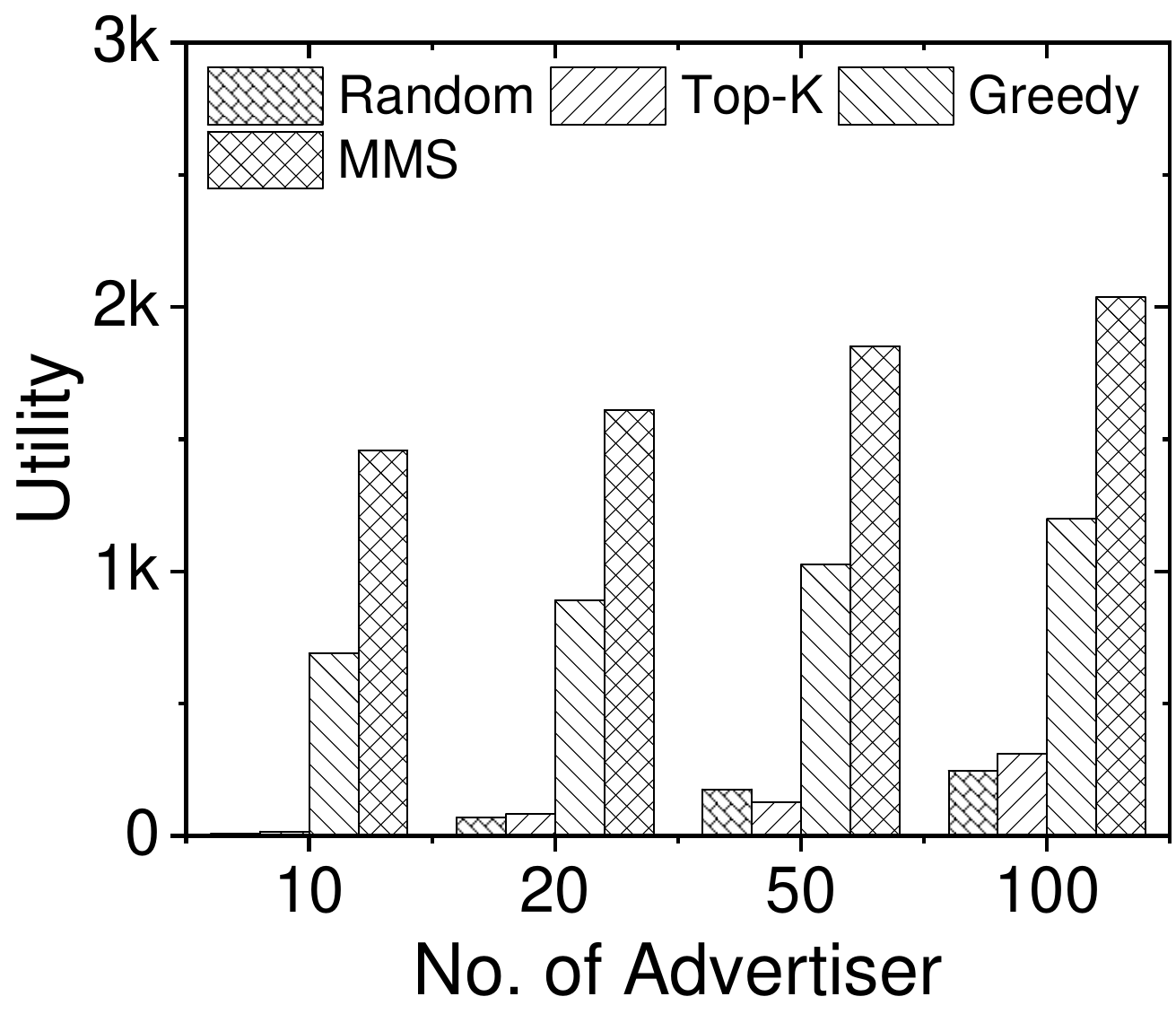} & \includegraphics[scale=0.175]{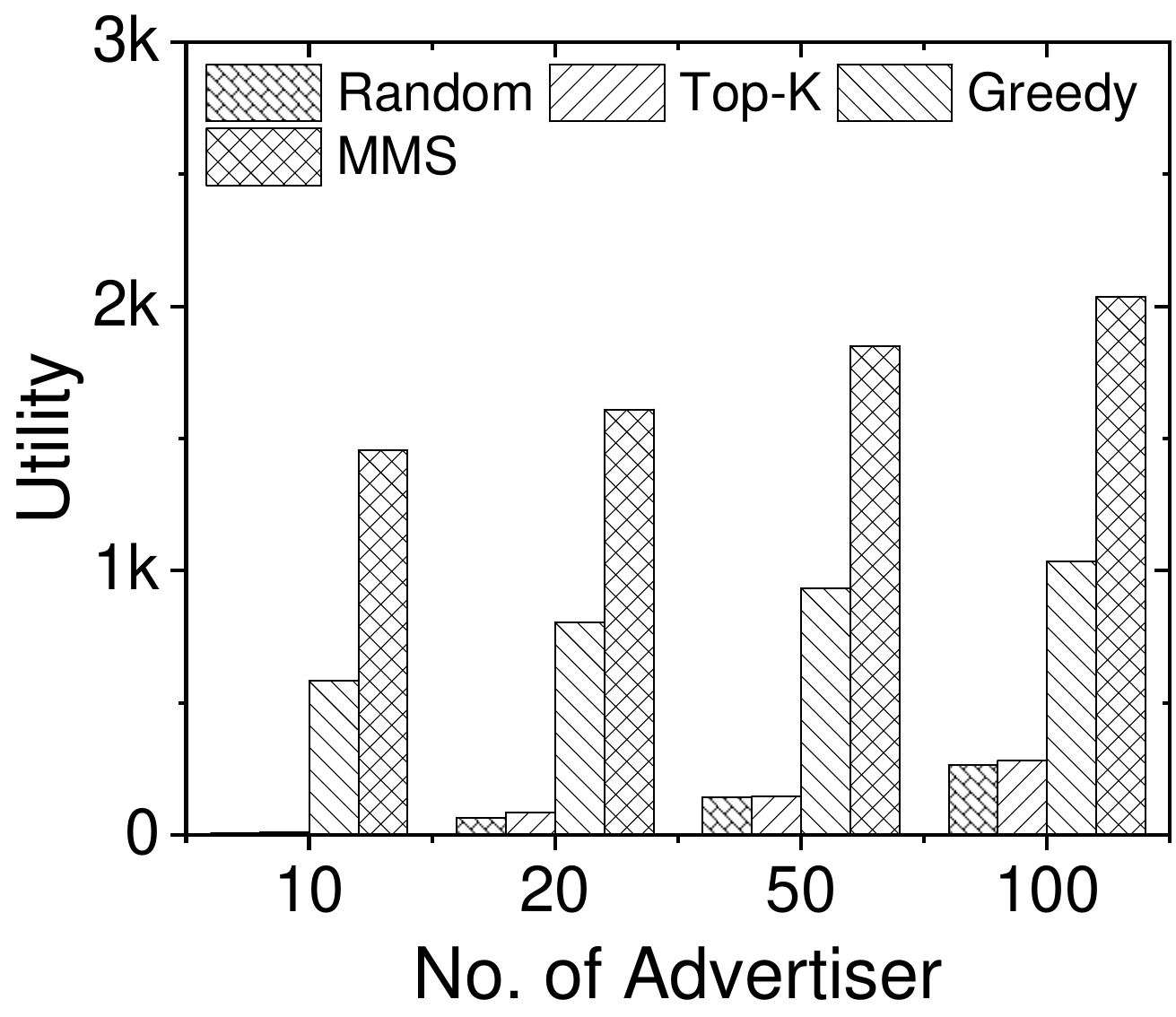} & \includegraphics[scale=0.175]{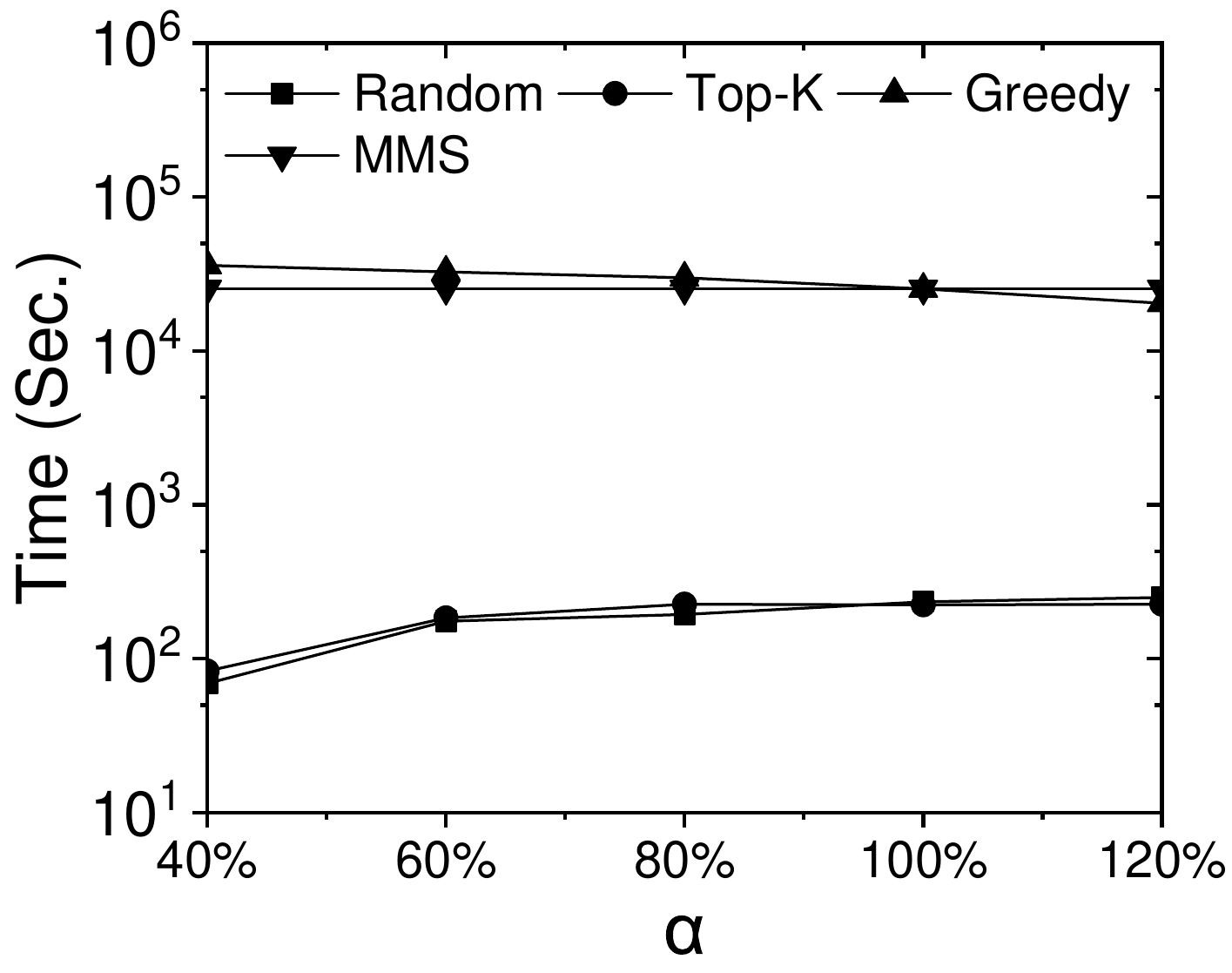}  \\
\tiny{(d)Varying $|\mathcal{A}|, \beta ~\text{with}~ \alpha = 100\%$} & \tiny{(e)  Varying $|\mathcal{A}|, \beta ~\text{with}~ \alpha = 120\%$}  & \tiny{(f) Runtime in NYC} \\
\includegraphics[scale=0.175]{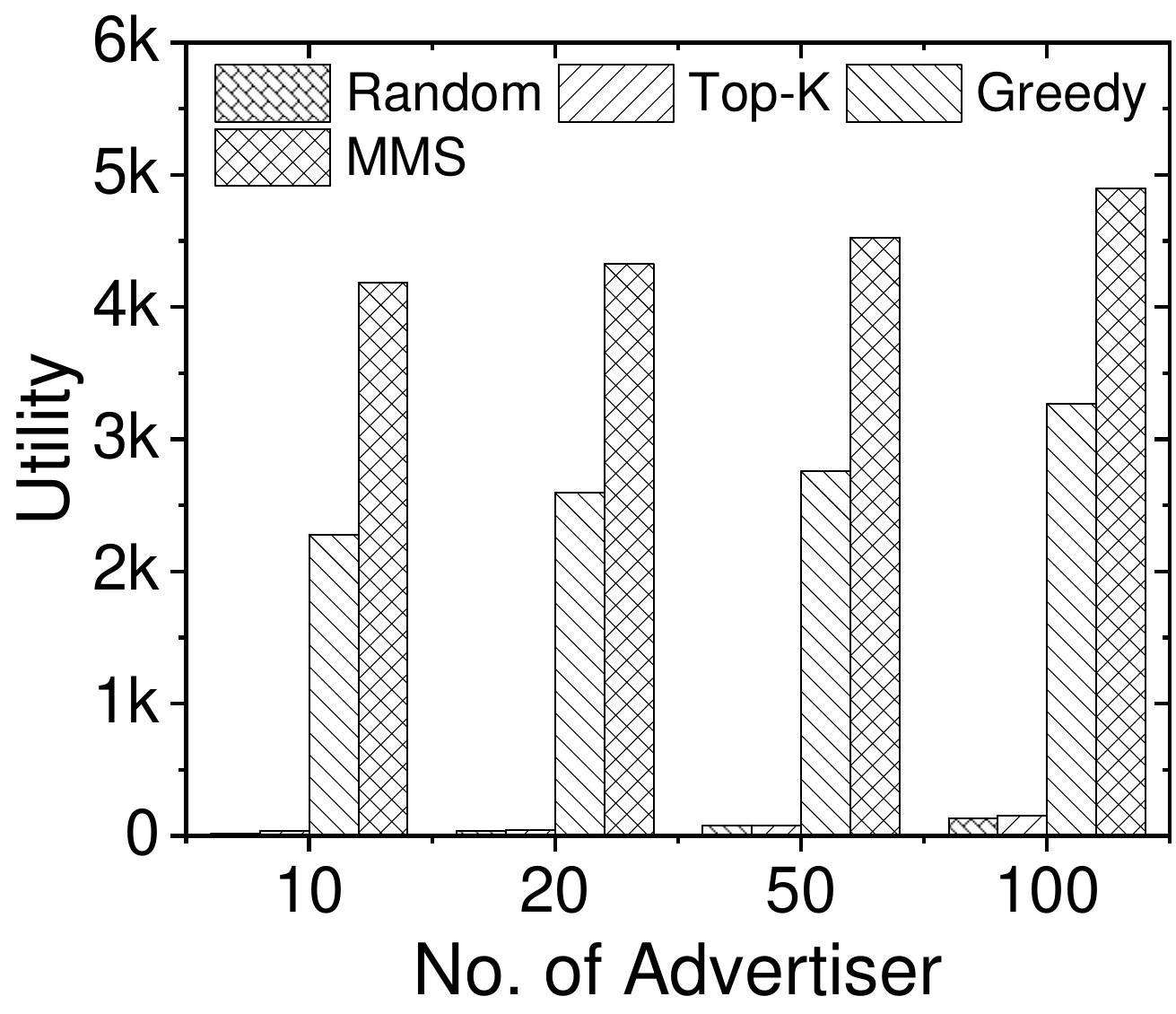} & \includegraphics[scale=0.175]{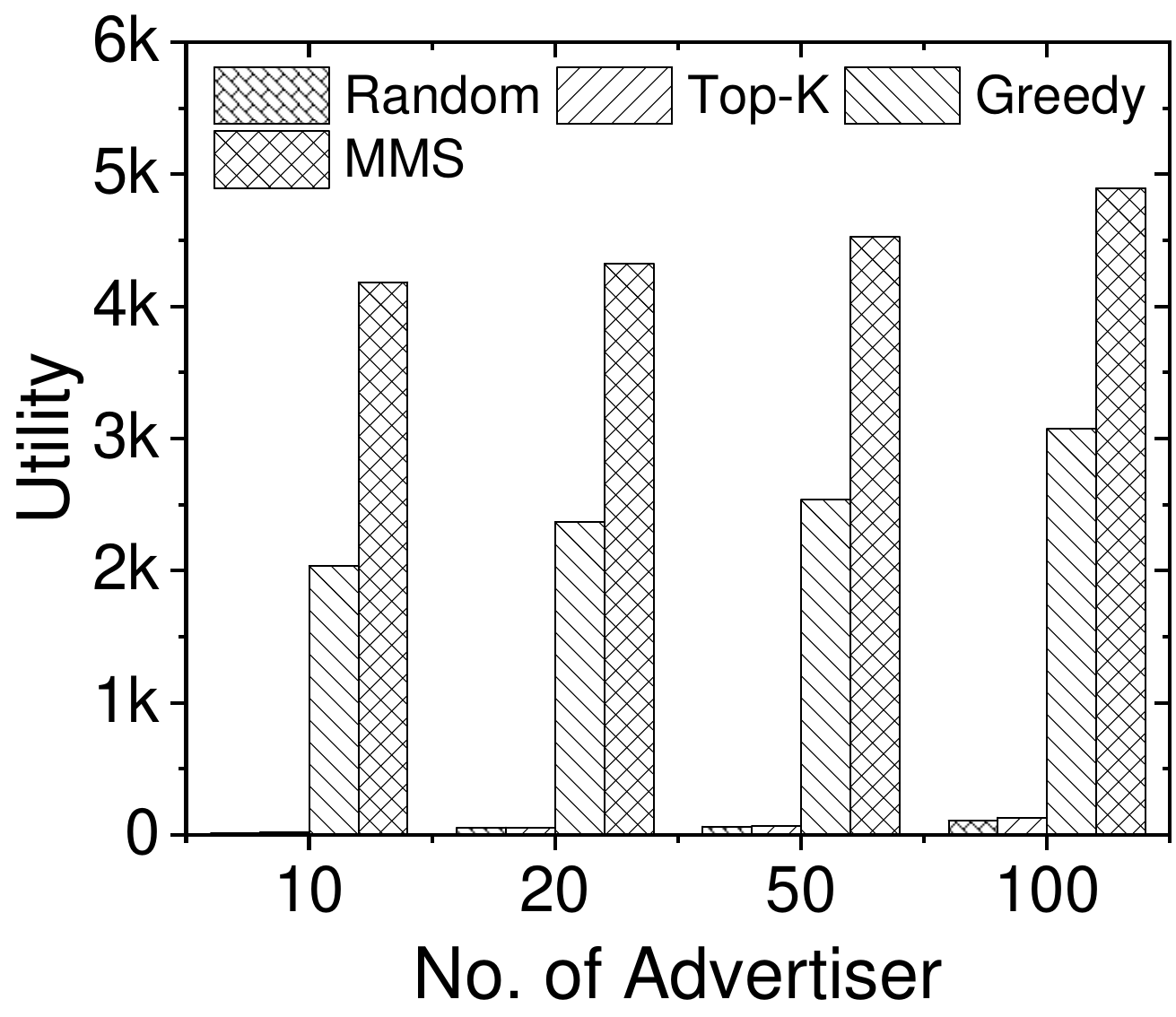}  & \includegraphics[scale=0.175]{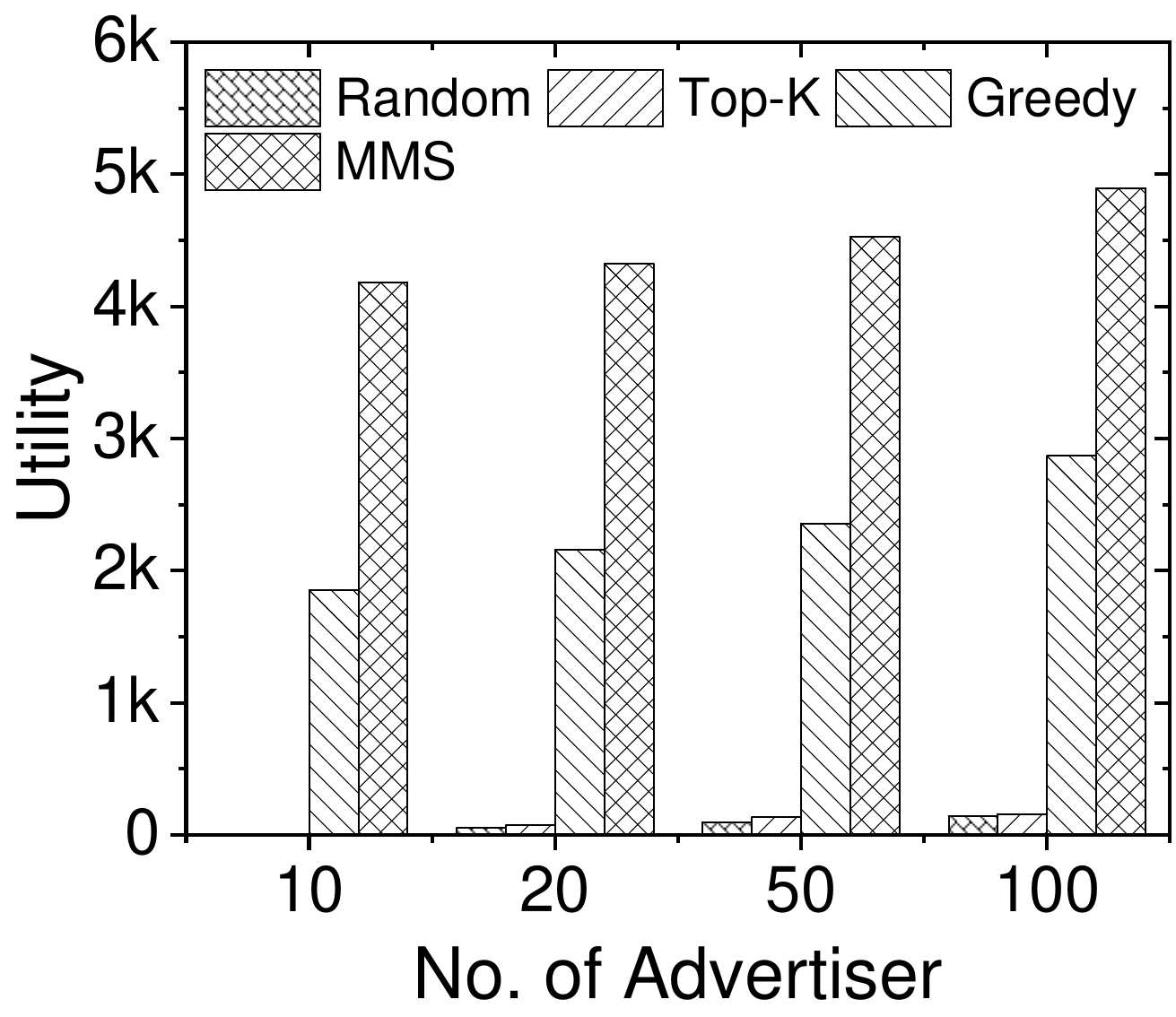} \\
\tiny{(g) Varying $|\mathcal{A}|, \beta ~\text{with}~ \alpha = 40\%$} & \tiny{(h)Varying $|\mathcal{A}|, \beta ~\text{with}~ \alpha = 60\%$} & \tiny{(i)Varying $|\mathcal{A}|, \beta ~\text{with}~ \alpha = 80\%$} \\
\includegraphics[scale=0.175]{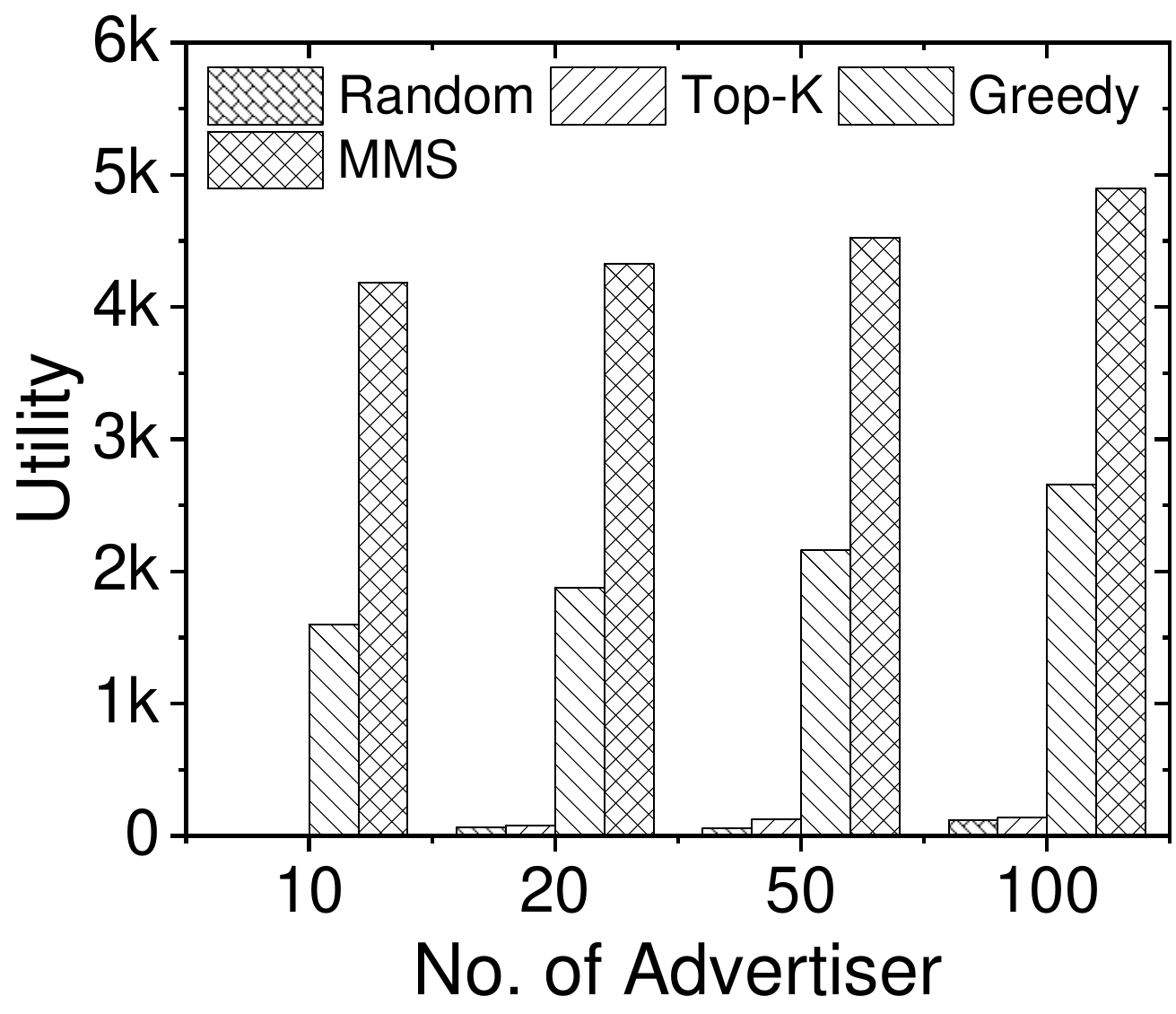} & \includegraphics[scale=0.175]{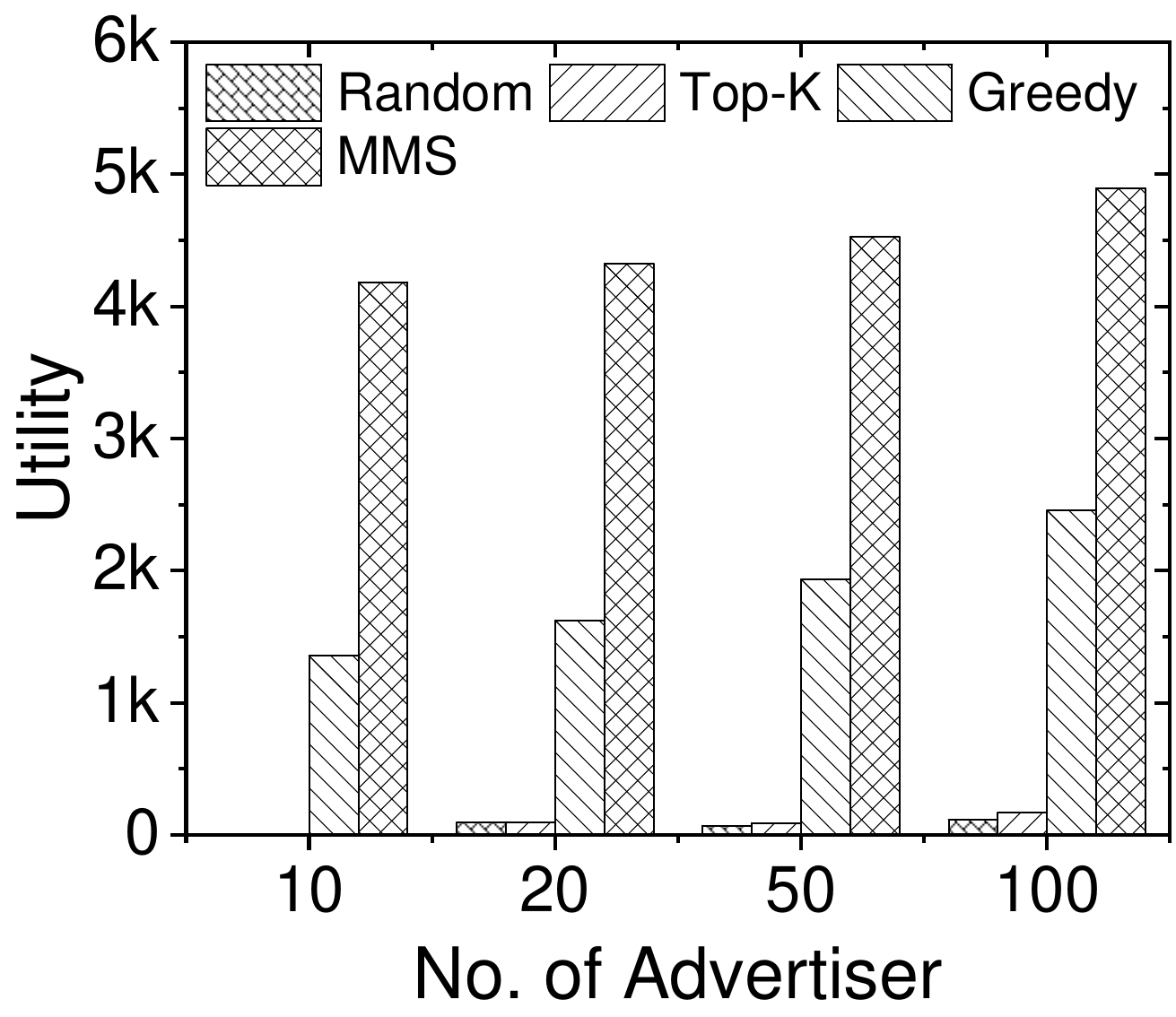} & \includegraphics[scale=0.175]{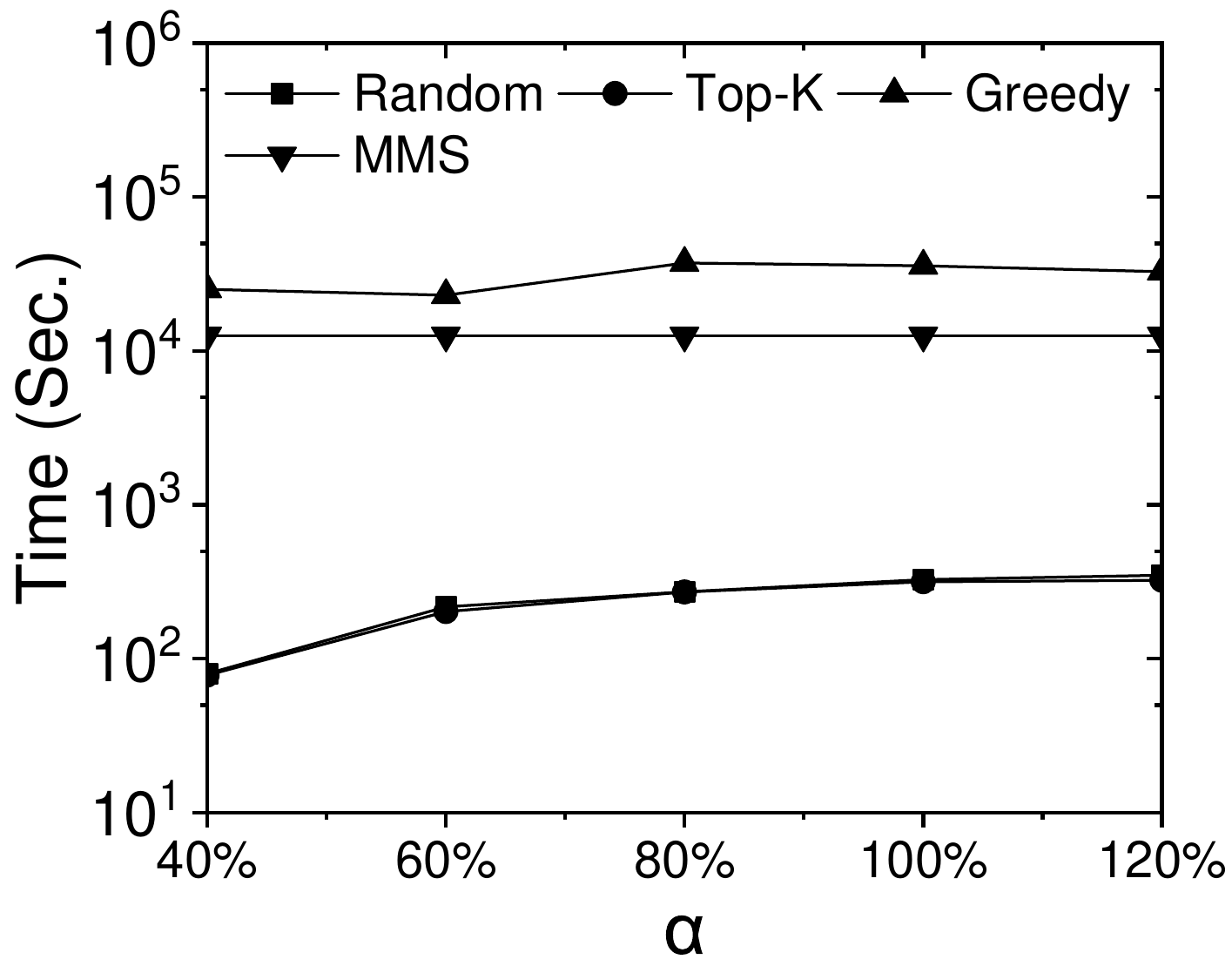} \\
\tiny{(j) Varying $|\mathcal{A}|, \beta ~\text{with}~ \alpha = 100\%$} & \tiny{(k) Varying $|\mathcal{A}|, \beta ~\text{with}~ \alpha = 120\%$} & \tiny{$(\ell)$ Runtime in LA} \\ 
\end{tabular}
\caption{ Varying $|\mathcal{A}|, \beta, \alpha$ in NYC $(a,b,c,d,e,f)$ and in LA $(g,h,i, j, k, \ell)$}
\label{Fig:Plot}
\end{figure*}
\vspace{-0.25cm}
\paragraph{Advertiser Vs. Utility.} Figure \ref{Fig:Plot} shows the effectiveness of the proposed solution approaches in the NYC and LA datasets, respectively. We have four main observations. First, with the fixed value of $\alpha$ and increasing number of advertisers from $10$ to $100$, the utility of advertisers increases. Second, with increasing $\alpha$ values from $40\%$ to $120\%$, the utility decreases because of a higher demand-supply ratio, $\alpha$, and in the extreme case when $\alpha = 120\%$ most of the advertisers influence demand not fulfilled. Third, among the baseline methods, the `Greedy' approach performs well compared to the `Random' and `Top-K'. We observe that the primary difference between the baseline and the proposed approximate maximin share approach is that when the influence provider runs out of slots i.e., $\alpha = 120\%$, some advertisers do not allocate any slots. On the other hand, the approximate Maximin share approach provides a balance allocation among all the advertisers, and no advertisers remain unallocated. Fourth, the greedy approach satisfies almost $75\%$ to $85\%$ advertisers individual influence demand while in maximin share approach $60\%$ to $70\%$ when $\alpha \geq 100\%$. 
\paragraph{Efficiency Test.}
Figure \ref{Fig:Plot}(f) and \ref{Fig:Plot}($\ell$) show the efficiency of the proposed as well as the baseline methods. We have three main observations. First, with the increase in the number of advertisers, the runtime increases for all the baseline methods. However, in the approximate maximin approach, run time decreases with the increase of the number of advertisers. Second, with the increase of $\alpha$ from $40\%$ to $120\%$ with fixed $\beta = 1\%$, the runtime decreases around $20\%$ to $25\%$ for both NYC and LA datasets. Third, the `Greedy' takes a much higher runtime compared to `Top-k' and `Random' due to an excessive number of marginal gain computations. In Figure \ref{Fig:Plot}, we only report the runtime for the default parameter settings, i.e., varying $\alpha$, $\beta = 5\%$, $|\mathcal{A}| = 20$ for both NYC and LA datasets. 

\paragraph{Scalability Test.} To show the scalability of the proposed methods, we vary $|\mathcal{A}|$ from $10$ to $100$ and $\alpha$ from $40\%$ to $120\%$. The experimental results show that the proposed approach is very sensitive compared to the baseline methods with fixed $\alpha, \beta,$ and$\epsilon$ values. Figure \ref{Fig:Plot} shows the impact of $\epsilon$ in the approximate maximin approach. The $\epsilon$ decides the sampling of billboard slots in Algorithm \ref{alg:RR_submodular}, varying from $0.1$ to $0.9$. It is observed that with the increase of $\epsilon$ value, the quality of the solutions and runtime decreases.
% \vspace{-0.1in}
\paragraph{Additional Discussions.} The additional parameters used in our experiments are $\epsilon$, $\gamma$, and $\theta$. First, the $\epsilon$ decides the sample set size in Algorithm \ref{alg:RR_submodular} and impacts runtime reduction. Second, the $\gamma$ is the penalty ratio, which controls the penalty imposed on the influence provider when the advertisers are unsatisfied. Third, the $\theta$ is the distance parameter that signifies a range in which a billboard slot can influence a user. With the increase of $\theta$ from $25m$ to $150m$, the influence of all the baseline and the proposed methods increases. This happens because one slot can influence the number of users within its range. In our experiments, we use $\epsilon = 0.3$, $\gamma = 0.5$, and $\theta = 100m$ as the default setting for both datasets. We have experimented with different values of $\epsilon$, $\gamma$, and $\theta$; however, due to space limitations, all are not reported.

\section{Concluding Remarks}\label{Sec:CFD}
In this paper, we have studied the Fair Billboard Slot Allocation Problem where given a set of billboard and advertiser databases, the influence demand and budget of the advertisers, the goal is to allocate slots to the advertisers fairly and efficiently. Considering maximin fair share as a fairness criterion, we have proposed a solution approach that leads to the approximate fair division of the slots. The proposed methodology has been analyzed to understand its time and space requirements and the performance guarantee. The experimental results on real-life trajectory and billboard datasets show the effectiveness and efficiency of the proposed solution approach. Our future study on this problem will remain concentrated on exploring other fairness notions and developing solution approaches that generate allocation with a guarantee of these fairness criteria.

% \begin{credits}
% \subsubsection{\ackname} This work is supported by the Start-Up Research Grant provided by the Indian Institute of Technology Jammu, India (Grant No.: SG100047).
% \end{credits}
\bibliographystyle{splncs04}
\bibliography{Paper}

\begin{thebibliography}{10}
\providecommand{\url}[1]{\texttt{#1}}
\providecommand{\urlprefix}{URL }
\providecommand{\doi}[1]{https://doi.org/#1}

\bibitem{ali2022influential}
Ali, D., Banerjee, S., Prasad, Y.: Influential billboard slot selection using
  pruned submodularity graph. In: International Conference on Advanced Data
  Mining and Applications. pp. 216--230. Springer (2022)

\bibitem{ali2023influential}
Ali, D., Banerjee, S., Prasad, Y.: Influential billboard slot selection using
  spatial clustering and pruned submodularity graph (2023)

\bibitem{ali2024minimizing}
Ali, D., Banerjee, S., Prasad, Y.: Minimizing regret in billboard advertisement
  under zonal influence constraint (2024)

\bibitem{10.1145/3605098.3636052}
Ali, D., Banerjee, S., Prasad, Y.: Regret minimization in billboard
  advertisement under zonal influence constraint. In: Proceedings of the 39th
  ACM/SIGAPP Symposium on Applied Computing. p. 329–336. SAC '24, Association
  for Computing Machinery, New York, NY, USA (2024)

\bibitem{10.1145/3033274.3085136}
Barman, S., Krishna~Murthy, S.K.: Approximation algorithms for maximin fair
  division. In: Proceedings of the 2017 ACM Conference on Economics and
  Computation. p. 647–664. EC '17, Association for Computing Machinery, New
  York, NY, USA (2017)

\bibitem{barman2018finding}
Barman, S., Krishnamurthy, S.K., Vaish, R.: Finding fair and efficient
  allocations. In: Proceedings of the 2018 ACM Conference on Economics and
  Computation. pp. 557--574 (2018)

\bibitem{10.1145/3495159}
Wang, L., Yu, Z., Guo, B., Yang, D., Ma, L., Liu, Z., Xiong, F.: Data-driven
  targeted advertising recommendation system for outdoor billboard. ACM Trans.
  Intell. Syst. Technol.  \textbf{13}(2) (jan 2022)

\bibitem{zhang2018trajectory}
Zhang, P., Bao, Z., Li, Y., Li, G., Zhang, Y., Peng, Z.: Trajectory-driven
  influential billboard placement. In: Proceedings of the 24th ACM SIGKDD
  international conference on knowledge discovery \& data mining. pp.
  2748--2757 (2018)

\bibitem{zhang2020towards}
Zhang, P., Bao, Z., Li, Y., Li, G., Zhang, Y., Peng, Z.: Towards an optimal
  outdoor advertising placement: when a budget constraint meets moving
  trajectories. ACM Transactions on Knowledge Discovery from Data (TKDD)
  \textbf{14}(5),  1--32 (2020)

\bibitem{zhang2019optimizing}
Zhang, Y., Li, Y., Bao, Z., Mo, S., Zhang, P.: Optimizing impression counts for
  outdoor advertising. In: Proceedings of the 25th ACM SIGKDD international
  conference on knowledge discovery \& data mining. pp. 1205--1215 (2019)

\bibitem{zhang2021minimizing}
Zhang, Y., Li, Y., Bao, Z., Zheng, B., Jagadish, H.: Minimizing the regret of
  an influence provider. In: Proceedings of the 2021 International Conference
  on Management of Data. pp. 2115--2127 (2021)

\end{thebibliography}

\end{document}